\documentclass[conference, 10pt]{IEEEtran}
\IEEEoverridecommandlockouts

\usepackage{main_style}

\begin{acronym}
\acro{AP}{access point}
\acro{LT}{Luby Transform}
\acro{BP}{belief propagation}
\acro{i.i.d.}{independent and identically distributed}
\acro{IoT}{Internet of Things}
\acro{mMTC}{massive machine type communication}
\acro{PER}{packet error rate}
\acro{SIC}{successive interference cancellation}
\acro{URLLC}{ultra-reliable low latency communication}
\acro{PMF}{probability mass function}
\acro{CRI}{collision resolution interval}
\acro{PGF}{probability generating function}
\acro{CPGF}{conditional probability generating function}
\acro{rhs}{right-hand side}
\acro{CSA}{coded slotted ALOHA}
\acro{MPR}{multi-packet reception}
\acro{MTA}{modified binary-tree algorithm}
\acro{SICTA}{binary-tree algorithm with SIC}
\acro{BTA}{binary tree-algorithm}
\acro{FCFS}{First Come First Served}
\acro{r.v.}{random variable}
\acro{w.r.t.}{with respect to}
\acro{RA}{random access}
\acro{TA}{tree algorithm} 
\acro{AoI}{Age of Information}
\acro{PT}{Policy Tree}
\acro{RL}{Reinforcement Learning}
\acro{DRL}{Deep Reinforcement Learning}
\acro{MARL}{multi-user reinforcement learning}
\acro{M2M}{machine to machine}
\acro{COTS}{Commercial off the shelf}
\acro{BS}{base station}
\acro{BF}{balance factor}
\acro{FBT}{full binary tree}
\acro{SA}{Slotted ALOHA}
\acro{MAC}{medium access and control}
\acro{MTC}{Machine Type Communication}
\acro{ADRA}{Age-Dependent Random Access}
\end{acronym}

\newcommand{\nint}{n}

\usepackage[absolute,showboxes]{textpos}
 
\TPMargin{5pt}
 
\newcommand{\copyrightstatement}{
    \begin{textblock}{0.84}(0.08,0.01)    
         \noindent
         \footnotesize
         \copyright 2022 IEEE. Personal use of this material is permitted. Permission from IEEE must be obtained for all other uses, in any current or future media, including reprinting/republishing this material for advertising or promotional purposes, creating new collective works, for resale or redistribution to servers or lists, or reuse of any copyrighted component of this work in other works.
    \end{textblock}
}

\begin{document}


\copyrightstatement

\title{Improving AoI via Learning-based Distributed MAC in Wireless Networks\\
}

\author{\IEEEauthorblockN{Yash Deshpande, Onur Ayan,  Wolfgang Kellerer\\
\IEEEauthorblockA{Chair of Communication Networks \\ Technical University of Munich, Germany \\
Email: \{yash.deshpande, onur.ayan, wolfgang.kellerer\}@tum.de }}}

\maketitle


\begin{abstract}
In this work, we consider a remote monitoring scenario in which multiple sensors share a wireless channel to deliver their status updates to a process monitor via an \ac{AP}. Moreover, we consider that the sensors randomly arrive and depart from the network as they become active and inactive. The goal of the sensors is to devise a medium access strategy to collectively minimize the long-term mean network \ac{AoI} of their respective processes at the remote monitor. For this purpose, we propose specific modifications to ALOHA-QT algorithm, a distributed medium access algorithm that employs a \ac{PT} and \ac{RL} to achieve high throughput. We provide the upper bound on the mean network \ac{AoI} for the proposed algorithm along with pointers for selecting its key parameter. The results reveal that the proposed algorithm reduces mean network \ac{AoI} by more than 50 percent for state of the art stationary randomized policies while successfully adjusting to a changing number of active users in the network. The algorithm needs less memory and computation than ALOHA-QT while performing better in terms of \ac{AoI}. 
\end{abstract}


\section{Introduction}
\label{sec:Introduction}

Applications involving \ac{IoT} have emerged across many industries to make up industry 4.0. In the near future, connected robotics and autonomous systems will be a significant driving force behind the design of 5G and beyond \cite{saadvision6g}. Monitoring the states of the robotic machinery and its environment via multiple sensors will result in a large amount of \ac{MTC} data. This data is characterized by periodic traffic generation, and short packet duration \cite{5gtrafficsurvey}. Many industry 4.0 applications such as factory robots, automated forklifts, and conveyor belts need not be active at all times. As individual tasks arise sporadically and are completed by the machines, the number of active users transmitting \ac{MTC} data will be dynamic. 

\ac{AoI} is a performance metric especially suitable for real-time monitoring applications because it measures the freshness of information coming from a remote source \cite{yatesaoisurevy}. \ac{AoI} depends on different aspects of the overall system, such as sampling rate of sensors, queue management, etc. From a \ac{MAC} design perspective, improving \ac{AoI} requires us to jointly optimize the transmission rate, delay, and the probability of successful reception of the information. Hence, adapting and designing special \ac{MAC} protocols with the goal of reducing \ac{AoI} have to be considered. 

One way to optimize \ac{AoI} in wireless \ac{IoT} networks is to use grant-based channel access protocols where a centralized scheduler keeps track of all the active users in the network and distributes the network resources efficiently to all the systems. However, they are typically inefficient for \ac{MTC} applications and complex to implement due to the overhead in signaling, and coordination \cite{mtcsurvey}. This reason compels us to look at simpler distributed grant-free \ac{RA} protocols for \ac{MTC} applications which are descendants of the well-known ALOHA \cite{Abramson:ALOHA} and \ac{SA} \cite{roberts1975aloha}. The simplicity in implementation of \ac{RA} protocols comes with the trade-off of poor performance in terms of \ac{AoI} when compared to the grant-based solutions due to the frequent collision of packets \cite{yatesraaoi}. In fact, it was shown in \cite{faraziaoi} that the mean network \ac{AoI} gap between grant based policies and grant free policies is $\mathcal{O}(n)$ where $n$ is the number of active users in the network. In order to bridge this gap, the users must overcome the collision problem by learning to coordinate and select transmission times in a way that the chances of packets colliding is reduced. 
 
Stationary randomized policies to reduce collisions have been suggested to improve the performance of \ac{RA} schemes. They rely on knowing or estimating the active number of users in the network \cite{atabay_threshold, roberts1975aloha, chenAoI, ChenIoT}. This value is difficult for the users to evaluate in a decentralized setup \cite{Vilgelm2021}.    
One such scheme \cite{chenAoI}, relies on Poisson distributed packet arrivals at the users to estimate the number of active users in the network in order to optimize an \ac{AoI} threshold. It however, cannot be applied to the \textit{generate-at-will} \cite{sungenerateatwill} model addressed in this paper. In \cite{ChenIoT} the authors present \ac{ADRA}, an extension of \cite{atabay_threshold} where the users access the channel only if a predefined \ac{AoI} threshold is exceeded. The channel access probability and the \ac{AoI} threshold are both a function of the number of active users in the network.
The DRR algorithm \cite{jiangdrr} achieves optimal \ac{AoI} by requiring the \ac{AP} to establish the number of active users in the network and relaying this information to the users via a feedback. This scheme therefore requires a more complex feedback as well as offloads some complexity to the \ac{AP} and moves in the direction of centralized scheduling. ALOHA-Q \cite{chu2015application} and ALOHA-QT \cite{alfaropt} require neither the users nor the \ac{AP} to ascertain the number of active users. Both the algorithms maintain the simplicity of classical \ac{RA} and achieve better channel utilization using \ac{RL} to coordinate with each other over the feedback of the \ac{AP}. The high utilization and flexibility of the ALOHA-QT algorithm piques our interest to investigate its performance in terms of \ac{AoI}. ALOHA-QT employs a \ac{PT} \cite{zhangPT} to divide the transmission slots in a frame into non-conflicting schedules which can be selected by the users. 

This paper proposes a better performing and computationally cheaper version of ALOHA-QT and calls it \textit{modified} ALOHA-QT or mAQT. The modifications are better suited for remote monitoring in \ac{MTC} applications where the number of active users is changing over time. We obtain an upper bound for the mean network \ac{AoI} by using properties of a well-known abstract data structure called the \ac{FBT}.

\begin{figure}
    \centering
    \includegraphics[width=0.30\textwidth]{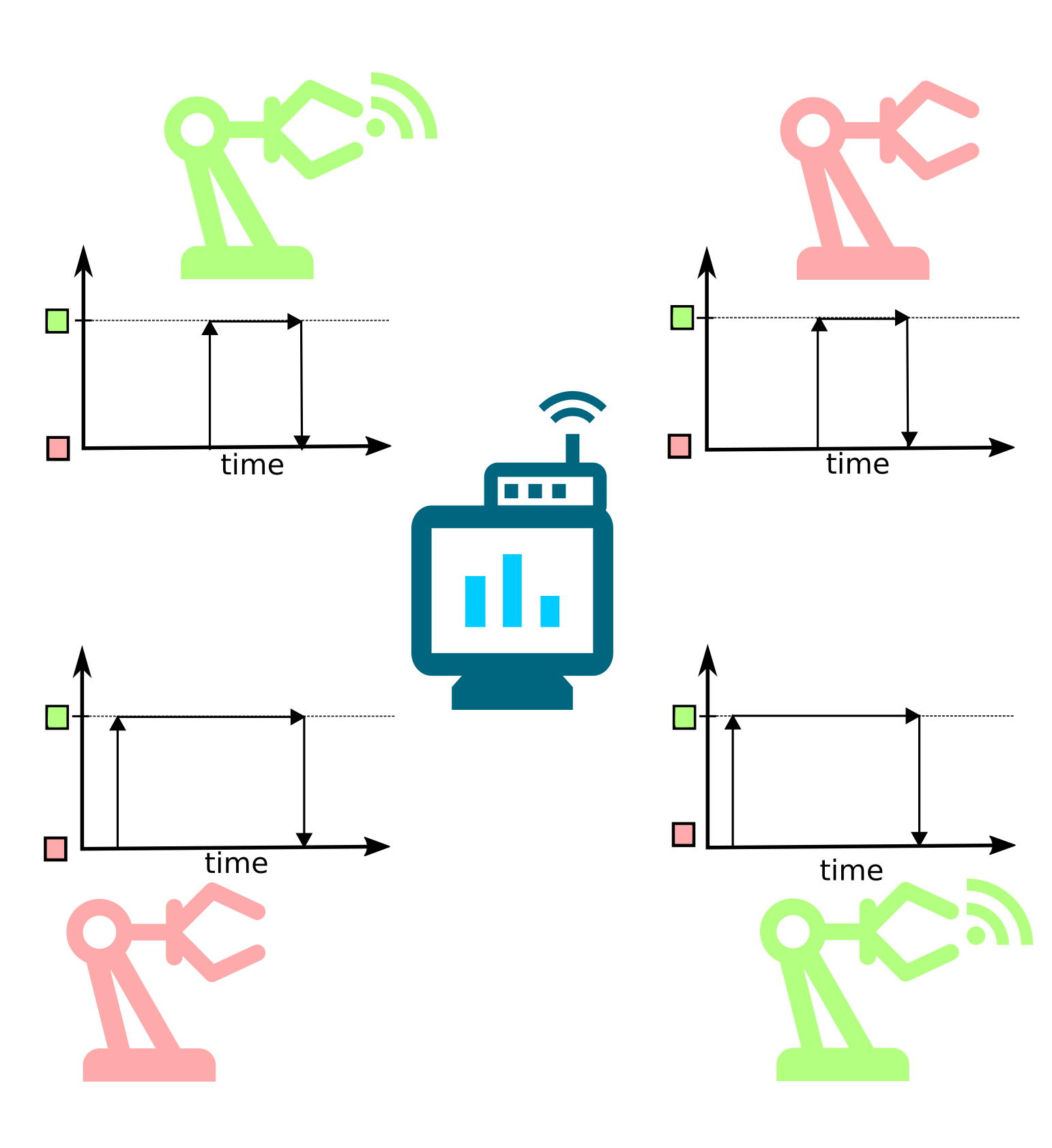}
    \caption{Example Scenario: Remote monitoring of factory robots who need to transmit their status updates only when they have a task at hand.}
    \label{fig:scenario_fig}
\end{figure}

\section{System Model}
\label{sec:System_Model}
We consider a remote process monitor connected to an \ac{AP}, receiving the status update packets of $M$ physical processes over a wireless network. Each process has a sensor and transmitter associated with it, and we call this subsystem a user to be consistent with the terminology used in \ac{MAC} protocols. Time is divided into slots of equal duration. The length of a status update packet is assumed to be constant, and the transmitter takes the duration of the entire slot to transmit a single packet. Throughout the paper, we express all time-related quantities in terms of slots. 

We consider that every user has two states - active and inactive. There are $\nint[t]$ active users in the network at time $t$. The number of slots a user spends in a state follows a geometric distribution with a transition probability of $p$ on every slot. Therefore, the average number of slots before a state transition for every user is $1/p = k$. On average, there is an activation or a deactivation once every $\frac{k}{M}$ slots in the entire network.  

To generate a packet, the sensor accurately samples the state of the process, and the transmitter encapsulates it into a status update packet ready for transmission. This packet generation process of a user is assumed to be \textit{generate-at-will} \cite{sungenerateatwill}, where an active user generates a new status update packet at the beginning of only those slots where it has decided to transmit. Hence, the user always has the freshest state encapsulated in any transmitted packet. The state of the physical process needs to be monitored only if the user is active. Such as system model can be imagined in a factory-like scenario shown in Figure \ref{fig:scenario_fig}, where the state of the machinery needs to be monitored. Individual machines are not occupied at all times and therefore only need to be monitored when they are performing a task. 

A decision to transmit $d_i[t] \in \{1,0\}$ is made at the start of slot $t$ by each active user $i$ according to its policy where $d_i[t]=1$ if the user decides to transmit and $d_i[t]=0$ if it decides to abstain from transmitting. At the end of the slot, the \ac{AP} broadcasts the slot outcome as feedback $F[t]\in \{1,e,0\}$ to all the users. If only one user transmits on the channel in a slot, the \ac{AP} is able to receive the packet successfully, and we call this a success slot, i.e. $F[t]=1$. When more than one user transmits in a slot, their signals interfere, and the \ac{AP} can neither decode any of the packets nor extract the number of users who transmitted on the channel. This scenario is called a collision, i.e. $F[t]=e$. We assume only an interference-limited channel such that transmission by any user fails only in the case of a collision. When no user transmits on the channel in a given time slot, then we say it is an idle slot, i.e. $F[t]=0$. 
The feedback is assumed to be immediate and perfect i.e., all active users receive the feedback at the end of the slot.  
%

At the beginning of a slot, the \ac{AP} sends the successfully received packet (if any) from the previous slot to the process monitor, which updates the state of the respective process accordingly. Hence, the \ac{AoI} of user $i$ at the process monitor is given by,

\begin{align}
\Delta_{i}[t+1] = \begin{cases}
 1, &\text{ if } d_{i}[t] = 1 \text{ and } F[t] = 1 \\
 \Delta_{i}[t] + 1, & \text{ otherwise }.
\end{cases}
\label{eqn:align}
\end{align}
An example of the evolution of \ac{AoI} for a particular user $i$ is shown in Figure \ref{fig:aoi_progression}. Here, $d_{i}[1]=1$, $F[1]=1$ $d_{i}[5]=1$ and $F[5]=1$. For time slots other than $1$ and $5$, the  user $i$ either refrained from transmitting or experienced collisions.  

\begin{figure}
    \centering
\begin{tikzpicture}

\begin{axis}[
width={\columnwidth}, height={4cm},
tick align=outside,
tick pos=left,
x grid style={white!69.0196078431373!black},
xlabel={$t$},
xmajorgrids,
xmin=-0.35, xmax=7.35,
xtick style={color=black},
xtick={0,1,2,3,4,5,6,7},
xticklabels={1,2,3,4,5,6,7,8},
y grid style={white!69.0196078431373!black},
ylabel={$\Delta_i$},
ylabel style={yshift=-1em},
ymajorgrids,
ymin=0.75, ymax=4.25
]
\path [draw=black, semithick]
(axis cs:0,1)
--(axis cs:0,6);

\path [draw=black, semithick]
(axis cs:4,1)
--(axis cs:4,6);

\addplot [ultra thick, FillColor0!, mark=square]
table {%
0 4
1 1
2 2
3 3
4 4
5 1
6 2
7 3
};
\end{axis}
\end{tikzpicture}
    \caption{Linear AoI progression: AoI of a user grows linearly until the user successfully transmits its status update packet at $t=1$ and $t=5$. The AoI becomes 1 after a successful transmission.} 
    \label{fig:aoi_progression}
\end{figure}
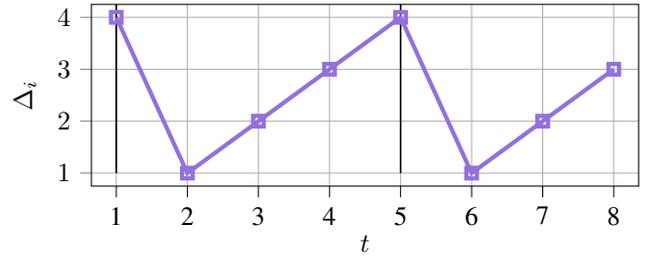

\section{Policy Tree Based Algorithm}
\label{sec:background}
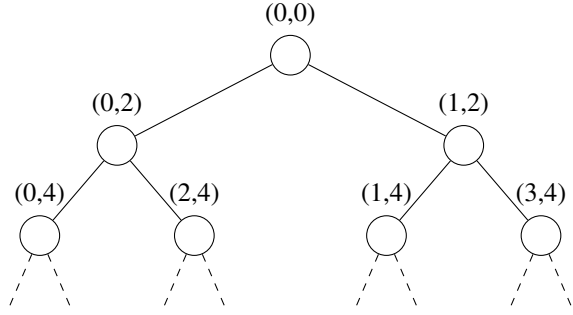
\begin{figure}
    \centering
    \tikzset{every tree node/.style={minimum width=1.5em,draw,circle},
         blank/.style={draw=none},
         edge from parent/.style=
         {draw,edge from parent path={(\tikzparentnode) -- (\tikzchildnode)}},
         level distance=1.2cm, level 1/.style={sibling distance=10mm},level 2/.style={sibling distance=5mm}, level 3/.style={sibling distance=5mm}}
\begin{tikzpicture}
\Tree
[.\node[label={(0,0)}]{}; 
    [. \node[label={(0,2)}]{};
    [ .\node[label={(0,4)}]{};
    \edge[dashed];\node[draw=none]{};
    \edge[dashed];\node[draw=none]{};
    ]
    [. \node[label={(2,4)}]{};
    \edge[dashed];\node[draw=none]{};
    \edge[dashed];\node[draw=none]{};
    ]
    ]
    [.\node[label={{(1,2)}}]{};
    [. \node[label={(1,4)}]{};
    \edge[dashed]; \node[draw=none]{};
    \edge[dashed]; \node[draw=none]{};
    ]
    [.\node[label={(3,4)}]{};
    \edge[dashed]; \node[draw=none]{};
    \edge[dashed]; \node[draw=none]{};
    ]
    ]
]
\end{tikzpicture}
    \caption{Policy Tree: Each node in the tree is called a schedule. Non-conflicting slots in a parent schedule are divided equally amongst the two children schedules.}
    \label{fig:policy_tree}
\end{figure}
A binary policy tree is shown in Figure \ref{fig:policy_tree}. Each node in the tree \footnote{Henceforth, we only talk about binary trees.} is represented by a tuple $(c,2^{l}) \mid 0 \leq c < 2^{l}$, $l \in \mathbb{N}$. The node is called a schedule with level $l$. Every active user in the network keeps a time slot counter $t$. The schedule $(c,2^{l})$ prescribes transmission if $t$ mod $2^{l} = c$. For example, the schedule $(3,4)$ prescribes a transmission when $t=3,7,11,15,19,23,\dots$. The tree is arranged in a way that the children of a parent schedule $(c,2^{l})$ are $(c, 2^{l+1})$ and $(c+2^{l}, 2^{l+1})$ and hence all slots that prescribe transmission for both the children are present in the parent. For example, while one child of the schedule $(3,4)$, e.g., $(3,8)$ prescribes transmission when $t=3,11,19,\dots$,  the other child  $(7,8)$ prescribes transmission when $t=7,15,23,\dots$. The schedules at the same level $l$ prescribe transmission at the same rate ($2^{l}$ slots) but with different offsets. In \ac{MAC} algorithms employing \ac{PT}, every user transmits according to one or more schedules. As long as users do not select ancestors or descendants of schedules selected by other users, they will have selected non-conflicting transmission slots. The number of schedules in the \ac{PT} is determined by the depth $J$ (maximum level) of the tree. A \ac{PT} with depth $J$ has $2^{J+1}-1$ schedules.   

\subsection{ALOHA-QT Algorithm}
\label{subsec:aloha_qt}


\begin{table}[]
\centering
\begin{tabular}{|c|c|cc|}
\hline
\multirow{2}{*}{Symbol} & \multirow{2}{*}{Parameter}   & \multicolumn{2}{c|}{Value}           \\ \cline{3-4} 
                        &                              & \multicolumn{1}{c|}{ALOHA-QT} & mAQT \\ \hline \hline
$J$                     & Depth of PT                  & \multicolumn{1}{c|}{6}       & 5    \\ \hline
$\eta$                  & Schedule selection threshold & \multicolumn{1}{c|}{0.95}     & -    \\ \hline
$\epsilon$              & Relinquishing probability    & \multicolumn{1}{c|}{0.02}     & -    \\ \hline
$\alpha^{+}$            & Increment factor             & \multicolumn{1}{c|}{0.2}      & 0.2  \\ \hline
$\alpha^{-}$            & Decrement factor             & \multicolumn{1}{c|}{-0.5}     & -0.5 \\ \hline
$\gamma_0$              & Weight initialization bias   & \multicolumn{1}{c|}{0.1}      & 0.1  \\ \hline
$\gamma_1$              & Weight initialization noise  & \multicolumn{1}{c|}{1.8}      & 1.8  \\ \hline
$w_{init}$              & Weight initialization factor & \multicolumn{1}{c|}{0.25}     & 0.25 \\ \hline
\end{tabular}
\caption{Symbols for the different parameters used in ALOHA-QT and mAQT. The optimal values given in the last column were obtained after running a gridsearch. These were used or the evaluation in Section V.}
\label{tab:params}
\end{table}
\begin{table}[]
\centering
\begin{tabular}{|c|c|}
\hline
Symbol        & Name                                                                      \\ \hline
$\mathcal{S}$ & Set of all schedules in the \ac{PT}                      \\ \hline
$\mathcal{A}$ & Set of all active schedules                                               \\ \hline
$\mathcal{I}$ & Set of selected schedules i.e policy                                      \\ \hline
$\mathcal{W}$ & Set of all weights. One for each schedule in the \ac{PT} \\ \hline
$t_i$         & Time slot counter                                                         \\ \hline
$\kappa$         & Boolean showing if user is active                                                         \\ \hline
\end{tabular}
\caption{State variables used in ALOHA-QT and mAQT. The cardinality of the sets of variables are shown in the last column. They essentially }
\label{tab:state_variables}
\end{table}

\begin{algorithm}
\begin{algorithmic}[1]
\item {\textbf{Initialization}: 
\\ \tabit[0.8cm] $t_i \gets 0$.
\\ \tabit[0.8cm] $S \gets \{(c,2^l)\mid 0 \leq c < 2^{l},0\leq l \leq J\}$.
\\ \tabit[0.8cm] $\forall w_{(c,2^l)} \in \mathcal{S}: w_{(c,2^l)} \gets \frac{w_{init}}{\gamma_{1}^{l}}\cdot(1-\gamma_{0}+\gamma_{0}\cdot\mathcal{U}(0,1) \}$
\\ \tabit[0.8cm] $\mathcal{W} \gets \{w_{\sigma}\}_{\sigma \in \mathcal{S}}$} 
\\\textit{At every slot, do}
   \STATE \textbf{Active Schedule Update}: \\ \tabit[0.8cm] $\mathcal{A} \gets \{(c,2^{l})\in \mathcal{S} \mid t_{i} \mod 2^{l} = c \}$.
   \STATE \textbf{Schedule Selection}: \\ \tabit[0.8cm]  $\mathcal{I}\gets \{\argmax_{\sigma \in \mathcal{S}}\mathcal{W}\}{\cup}\{\sigma\in \mathcal{S}\mid w_{\sigma}>\eta\}$. 
   \STATE \textbf{Decision}: 
   \\\tabit[0.8cm]\textbf{if} $\mathcal{A}\cap\mathcal{I}\neq\varnothing$ and $\kappa = 1$ \textbf{then} $d_{i}[t_{i}]=1$ 
   \\\tabit[0.8cm]\textbf{else} $d_{i}[t_{i}]=0$.
   \STATE \textbf{Reward Selection}: 
   \\\tabit[0.8cm]\textbf{if} $(F[t_{i}], d_{i}[t_{i}])=(0,0)$ or $(1,1)$ \textbf{then} $\alpha \gets \alpha^{+}$
   \\\tabit[0.8cm]\textbf{else} $\alpha \gets \alpha^{-}$. 
   \STATE \textbf{Weight Update}: 
   \\ \tabit[0.8cm] $\forall \sigma \in \mathcal{A} : w_{\sigma}^{\prime} \gets w_{\sigma} \cdot e^{\alpha\cdot\mathcal{U}(0,1)}$. 
   \STATE \textbf{Voluntary Relinquishment}: 
   \\\tabit[0.8cm]\textbf{if} $\mathcal{U}(0,1) \leq \epsilon$ \textbf{then} $\forall\sigma \in\mathcal{A}:w_{\sigma}^{\prime}\gets 0$.
   \STATE \textbf{Weight Normalization}:
   \\\tabit[0.8cm] $W \gets \sum_{\sigma \in \mathcal{S}}w_{\sigma}$, $W^{\prime} \gets \sum_{\sigma \in \mathcal{S}}w_{\sigma}^{\prime}$
   \\\tabit[0.8cm] $\delta \gets W - W^{\prime}$
   \\\tabit[0.8cm] \textbf{if} $\delta > 0$ and $W^{\prime} < w_{init}\cdot|\mathcal{S}|$
   \\\tabit[0.8cm] \textbf{then} $\forall \sigma \in \mathcal{S}: X_{\sigma} \gets \mathcal{U}(0,1)$
   \\\tabit[1.6cm] $\forall \sigma \in \mathcal{S}: w_{\sigma} \gets w_{\sigma}^{\prime}+\delta\cdot(X_{\sigma}/\sum_{\sigma}X_{\sigma}) $
   \\\tabit[0.8cm] \textbf{else}  $\forall \sigma \in \mathcal{S}: w_{\sigma} \gets w_{\sigma}^{\prime}$ 
   \STATE \textbf{Bound Enforcement}: 
   \\\tabit[0.8cm]$\forall \sigma \in \mathcal{S}: w_{\sigma} \gets \min(1,w_{\sigma})$
   \STATE \textbf{Time Increment}: \\\tabit[0.8cm]$t_{i} \gets t_{i}+1$. 
\end{algorithmic}
\textit{Every instance of $\mathcal{U}(0,1)$ in the above algorithm is an independent random sample drawn from the uniform distribution in the interval $[0,1]$}.
\caption{The ALOHA-QT Algorithm}
\label{algo:aloha_qt}
\end{algorithm}

ALOHA-QT (Algorithm \ref{algo:aloha_qt}) is a distributed expert-based \ac{RL} algorithm using which, every user selects non-conflicting schedules in a \ac{PT}. Notation of parameters and state variables used in this algorithm are given in Tables \ref{tab:params} and \ref{tab:state_variables} respectively. The algorithm iteratively assigns a weight $w_{(c,2^{l})} \in [0,1]$ to every schedule $(c,2^{l})$ in the \ac{PT} according to its potency to achieve a non-conflicting transmission. First (in step 0), every user initializes the weights of all the schedules in the \ac{PT} such that higher schedules (closer to the root node of the \ac{PT}) have higher weights. This ensures that the users explore transmitting at higher rates before moving down the \ac{PT}. Small noise is added to each weight to reduce the probability of two schedules at the same level being initiated with the same weight. The noise also ensures that the initial behavior of all the users is not the same. The rest of the steps are then performed by the users once every slot. 

\begin{enumerate}
    \item \textbf{Step 1}: The user updates in a memory location all schedules which are active i.e., ones who prescribe a transmission in the current slot. At any time slot, there are always $J+1$ active schedules.
    \item \textbf{Step 2}: The user selects from the entire \ac{PT} all schedules with weights $>\eta$ in addition to the schedule with the maximum weight (we call this primary schedule).
    \item \textbf{Step 3}: The user transmits on the channel if it is active and any of the selected schedules suggest the user to transmit.
    \item \textbf{Step 4}: Feedback from the \ac{AP} such as the one mentioned in section \ref{sec:System_Model} is used to select a positive or negative reward by the user.
    \item \textbf{Step 5}: The weights of all $J+1$ active schedules are updated at every time slot. A negative reward selection in step 4 will decrease the weights while a positive reward selection will increase the weights in this step. As users become active and inactive, schedules become unfavourable and promising respectively. Therefore, a multiplicative update strategy is used for facilitating the quick adaptation of weights in such a dynamic environment. Small noise is added to all the updated weights to break ties between schedules that might have the same value of weight. From a classical \ac{RL} sense, this is the \textit{reward function} of the algorithm. 
    \item \textbf{Step 6}: With a small constant probability $\epsilon$, the user sets the weights of all active schedules to 0. This happens randomly once every few hundred slots to make sure that the users do not hold higher schedules indefinitely.
    \item \textbf{Step 7}: If the weights of active schedules were reduced (either due to negative feedback or relinquishment) and if the sum of the weights in the \ac{PT} falls below a value of $w_{init}\cdot|\mathcal{S}|$, the lost weights in this step are redistributed across all the schedules in the \ac{PT}. This allows the users to quickly explore alternative schedules if their selected schedule starts to give negative feedback \cite{herbster1998tracking}.
    \item \textbf{Step 8}: The users make sure that the weights of all schedules remain at most 1. This way, the positive multiplicative update from a good schedule does not increase indefinitely.
    \item \textbf{Step 9}: The time counter is updated so that the user can process the next slot.
\end{enumerate}

In this manner, the users explore schedules in the \ac{PT} and learn to coordinate over time in order to select non-conflicting schedules in a distributed manner. This coordination is achieved only via the broadcast feedback at the end of the slot.  
It is important to note that each user $i$ selects schedules in a distributed manner and the time slot counter $t_i$ does not need to be the same for all users in the network: a schedule $(c, 2^{l})$ for a user $a$ with time slot counter $t_a$ is the same as a schedule $((c+s)$ mod $2^{l}, 2^{l})$ for a user $b$ with time slot counter $t_b = t_a + s$. Thus, the users in the network need not synchronize their time slot counters. This property is especially useful if we need to change the number of users $M$. A new user can be introduced in the network and it needs to only synchronize the start of a time slot and does not need to obtain any additional information from other users or the \ac{AP}. 

\subsection{Application specific changes to ALOHA-QT}

The ALOHA-QT algorithm \cite{alfaropt} is designed for a system to optimize throughput in a fair manner by avoiding collisions via implicit coordination over the feedback. The authors show its applicability in a system model where all the users are active, when the number of active users is slowly increasing or slowly decreasing as well as when there is frequent activation and deactivation at the start of every 100th slot. The system model makes no assumptions on the maximum number of users that the scheme can accommodate. 

This differs from the model which we have defined in section \ref{sec:System_Model}. In this model, the activation and deactivation are less frequent but random and do not need to occur at the beginning of a slot batch. The system designer knows the maximum number of users $M$ in the network. We applied the ALOHA-QT algorithm to our system model and made two key observations in terms of \ac{AoI}. These observations lead us to suggest to following changes to the ALOHA-QT algorithm to better suit the system model presented in section \ref{sec:System_Model}.
\subsubsection{Skip voluntary relinquishment in step 6:}
This step was designed to make sure that no user holds a higher schedule for a long time. If any user relinquishes selected schedules in this step, other users compete to grab these schedules, causing collisions. At the same time, the user who relinquished the schedules begins transmitting in some other users selected schedule producing further collisions. This causes the throughput to drop temporarily. While this trade-off might be useful in maintaining fairness with a static number of users over a long time, it was found that it does not help when the users spend a random amount of time in the network before becoming inactive. Therefore, we suggest skipping this step for our application. 

\subsubsection{Select only one schedule in step 2:}
The possibility of allowing the users to select more than one schedule in ALOHA-QT was designed to allow a flexible \textit{throughput} for all active users. Firstly, it was observed that in most cases, the secondary schedules were either children or siblings of the primary selected schedule i.e. the one with maximum weight. Secondly, it was also observed that once a user selects one or more schedules, their weights quickly rise up to become 1. This is caused by the reinforcing effect of the multiplicative update on receiving positive feedback. A newly active node entering the network needs to compete with the selected schedules of many users in order to find a new collision-free schedule in the \ac{PT}. If we allow the users to select more than one schedule, the newly active user is likely to face competition from more schedules and hence take more time to find a new collision free-schedule in the \ac{PT}. 

\subsubsection{When PT is settled, only run Step 3:}
With the suggested modifications, the users select only one unique schedule in the \ac{PT} which is neither an ancestor nor a descendent of the schedules selected by other users. When we establish the network with all active users running the algorithm with the mentioned modifications, they take some time (we call this settling time) to select their unique schedule in the PT. This results in the network achieving full channel utilization. We call this condition a ``settled \ac{PT}''. However the tree does not remain settled forever, as there will be arrivals and departures of users in the network as they become active and inactive. These events unsettle the tree for a certain amount of time (we call this resettling time) before the tree settles once again. We propose that the users in a settled tree do not change their selected schedules unless there is an arrival or departure of a user in the network. Therefore, they do not need to update any state variables when they are in a settled tree. If there are no collisions or idles detected in the last $2^J$ slots, the users deem that the \ac{PT} is settled. In a settled state, the users perform only step 3 of the algorithm. Thus doing minimum work while still avoiding collisions and idle slots entirely.    
 
The comparison results between pure ALOHA-QT and the modifications (modified ALOHA-QT or mAQT) are shown in section \ref{sec:Evaluation}. 
\begin{figure*}[ht]
    \centering
    \begin{subfigure}[b]{0.48\textwidth}
    \centering
    \caption{}
    \tikzset{every tree node/.style={minimum width=1.5em,draw,circle},
         blank/.style={draw=none},
         edge from parent/.style=
         {draw,edge from parent path={(\tikzparentnode) -- (\tikzchildnode)}},
         level distance=0.9cm, level 1/.style={sibling distance=10mm},level 2/.style={sibling distance=5mm}, level 3/.style={sibling distance=5mm}}
\begin{tikzpicture}
\Tree
[.\node[label={(0,0)}]{}; 
    [. \node[label={(0,2)}]{};
    [ .\node[label={(0,4)}, fill=PlotColor2!40]{\small{1}};
    ]
    [. \node[label={(2,4)}]{};
    \edge[];\node[label={(2,8)}, fill=PlotColor1!40]{\small{4}};
    \edge[];\node[label={(6,8)}, fill=PlotColor3!40]{\small{5}};
    ]
    ]
    [.\node[label={{(1,2)}}]{};
    [. \node[label={(1,4)}, fill=PlotColor4!40]{\small{2}};
    ]
    [.\node[label={(3,4)}, fill=PlotColor5!40]{\small{3}};
    ]
    ]
]
\end{tikzpicture}
\vspace{0.70cm}
    \label{fig:settled_tree_ex_5}
    \hfill
    \end{subfigure}
    \begin{subfigure}[b]{0.48\textwidth}
    \centering
    \caption{}
\begin{tikzpicture}
\begin{axis}[
width={\columnwidth}, height={4.5cm},
legend cell align={left},
legend style={fill opacity=0.8, draw opacity=1, text opacity=1, draw=white!80!black},
tick align=outside,
tick pos=left,
x grid style={white!69.0196078431373!black},
xtick={0,5,10,15,20},
xticklabels={$t_{s}+0$,$t_{s}+5$,$t_{s}+10$,$t_{s}+15$,$t_{s}+20$},
xmajorgrids,
xmin=-1, xmax=21,
xtick style={color=black},
xlabel={$t$},
y grid style={white!69.0196078431373!black},
ymajorgrids,
ymin=0.65, ymax=8.35,
ylabel={$\Delta_i$},
ylabel style={yshift=-2em},
ytick style={color=black},
ytick={1,2,4,8},
]
\addplot [semithick, PlotColor2]
table {%
0 1
1 2
2 3
3 4
4 1
5 2
6 3
7 4
8 1
9 2
10 3
11 4
12 1
13 2
14 3
15 4
16 1
17 2
18 3
19 4
20 1
};
\addlegendentry{$i$=1}
\addplot [semithick, PlotColor1]
table {%
0 3
1 4
2 5
3 6
4 7
5 8
6 1
7 2
8 3
9 4
10 5
11 6
12 7
13 8
14 1
15 2
16 3
17 4
18 5
19 6
20 7
};
\addlegendentry{$i$=4}
\addplot [semithick, PlotColor3]
table {%
0 7
1 8
2 1
3 2
4 3
5 4
6 5
7 6
8 7
9 8
10 1
11 2
12 3
13 4
14 5
15 6
16 7
17 8
18 1
19 2
20 3
};
\addlegendentry{$i$=5}
\addplot [semithick, PlotColor4]
table {%
0 2
1 3
2 4
3 1
4 2
5 3
6 4
7 1
8 2
9 3
10 4
11 1
12 2
13 3
14 4
15 1
16 2
17 3
18 4
19 1
20 2
};
\addlegendentry{$i$=2}
\addplot [semithick, PlotColor5]
table {%
0 4
1 1
2 2
3 3
4 4
5 1
6 2
7 3
8 4
9 1
10 2
11 3
12 4
13 1
14 2
15 3
16 4
17 1
18 2
19 3
20 4
};
\addlegendentry{$i$=3}
\end{axis}

\end{tikzpicture}
    \label{fig:settled_tree_tx_pattern}
    \end{subfigure}
    \caption{Settled Policy Tree: (a) When a policy tree is settled, it represents a full binary tree where leaf nodes have no children and all other nodes have $2$ children. The resulting network attains a throughput of $1$. (b) Every user $i$ has max AoI equal to $2^{l_{i}}$ where $l_i$ is its selected level in a settled tree. $t_s$ is any time slot after the tree is settled.}
    \label{fig:settled_tree_full}
\end{figure*}
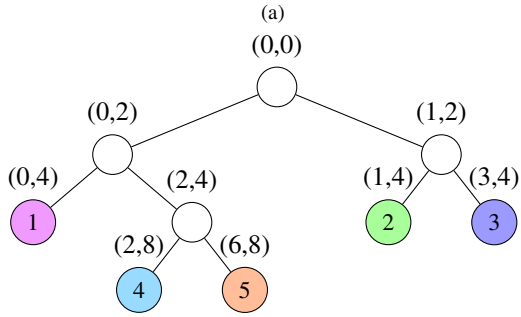
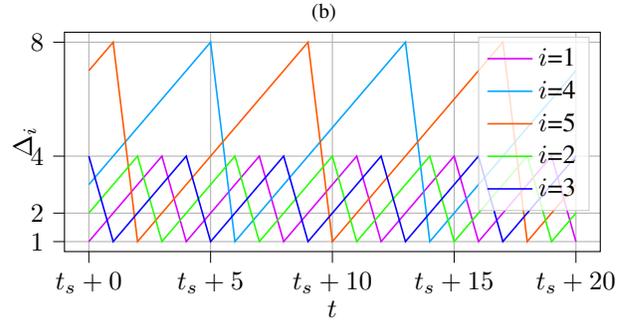
%

\section{Analysis}
\label{sec:Analysis}
We assume that the average time between two events that disturb a settled tree $\frac{k}{M}$, is less than the average resettling time. Hence the users spend more time in a settled tree rather than in an unsettled one. In this section, we obtain insights into the \ac{AoI} performance of the mAQT when the tree is settled. In such a case, the number of active users $\nint[t]$ for all time slots $t$ in this period is constant \footnote{In fact, users in mAQT do not even update their time slot counters when the tree is settled.}. Therefore, we drop the time index for the ease of notation and refer to the number of active users as $\nint$ in this section. 

%
One realization of a settled tree for $n=5$ is shown in Figure \ref{fig:settled_tree_ex_5}. The $n$ leaf nodes filled with color are the selected schedules by each of the 5 users in the system. The resulting \ac{AoI} of each user over 20 time slots after the tree is settled is shown in Figure \ref{fig:settled_tree_tx_pattern}. It can be seen that there is a successful transmission at every slot, resulting in the network achieving full channel utilization. 

The $\Delta_i$ for every user is cyclic with period $2^{l_{i}}$. The long-term $(t \to \infty)$ mean \ac{AoI} per slot of a user $i$, $\Bar{\Delta}_i$ in a settled tree is the mean \ac{AoI} of the period,
\begin{equation}
    \Bar{\Delta}_i = \frac{1}{2^{l_i}}\sum_{t=t_{0}+1}^{t_{0} + 2^{l_i}}\Delta_i [t] =\frac {\sum_{t=1}^{2^{l_i}}t}{2^{l_i}} = \frac{2^{l_i} + 1}{2}. 
    \label{eqn:mean_aoi_pt}
\end{equation}

The mean \textit{network} \ac{AoI} $\Bar{\Delta}$ is the mean \ac{AoI} per slot of all active users in the network,
\begin{equation}
    \Bar{\Delta} \triangleq \frac{1}{n}\sum_{i=1}^{n}\Bar{\Delta}_i =\frac{1}{2}(1 + \frac{1}{\nint}\sum_{i=1}^{\nint}2^{l_i}) .
    \label{eqn:aoi_for_network_scheudle}
\end{equation}

Even though the channel is fully utilized when the tree is settled, the fraction of the channel resources utilized by each user in the system can be different. Every user obtains ${2^{-l_i}}$-th of the channel resources and $\sum_{i=1}^{n}{2^{-l_{i}}} = 1$. The mAQT algorithm has a degree of randomness, which may lead to different realizations of the \ac{PT} for the same number of users $\nint$. Thus the users may obtain different values of $\Bar{\Delta}$ even for the same number of users $\nint$ depending on the manner in which they settle. To analyze this further we make use of the \textit{balance} property of a \ac{FBT}\cite{rosenalgorithm}, an abstract data structure commonly used in computer science. 

\subsection{Settled Trees as Full Binary Trees}
A \ac{FBT} is defined as a tree structure where every node has either two or no children. A settled \ac{PT} looks exactly like a \ac{FBT} with a node in the \ac{FBT} representing a schedule in the \ac{PT}. The leaf nodes (ones without children) represent the selected schedules of the $\nint$ users. Any sub-tree of a \ac{FBT} is also a \ac{FBT}. We define the set of selected levels (of schedules) by the users of particular realization $r$ of a settled \ac{PT} with $\nint$ users as $\mathcal{T}_r = \{l_1,l_2,...,l_n\}$. The height of a \ac{FBT} $h$ is defined as the number of edges between the root node and the farthest leaf node. For realization $r$ of settled tree the height is $l_{max}^{r} = \max(\mathcal{T}_r)$.        

A fully balanced \ac{FBT} is where the difference between heights of the two principle sub-trees of any sub-tree is at most 1. The tree in Figure \ref{fig:settled_tree_ex_5} is an example of a fully balanced \ac{FBT}. The closer the values of $l_i$  in a particular realization of a settled tree are to each other, the more \textit{balanced} that \ac{PT} is.

\begin{theorem}
\label{th:tree_balanceness}
The mean \ac{AoI} per user $\Bar{\Delta}$ of a settled policy tree depends on how balanced the tree is. For a given $\nint$, a fully balanced \ac{FBT} is $\Bar{\Delta}$ optimal.  
\end{theorem}
\begin{proof}
A \ac{FBT} has at least two sibling leaf nodes at height $h$. Let the leaf node(s) for realization $a$ at a higher level be at $l_{min}^{a}$ such that $l_{max}^{a} > l_{min}^{a}$. Now, consider the two siblings at $l_{max}^{a}$ and one leaf node at $l_{min}^{a}$. We explicitly write these three values in last sum in equation \eqref{eqn:aoi_for_network_scheudle}, 
\begin{equation}
    \Bar{\Delta}_{a} = \frac{1}{2}(1 + \frac{1}{\nint}\sum_{i=4}^{\nint}2^{l_i}+ \frac{1}{n}(2^{l_{max}^{a}}+2^{l_{max}^{a}}+2^{l_{min}^{a}})) ,
    \label{eqn:unbalanced_no_1}
\end{equation}
where $\Bar{\Delta}_{a}$ stands for the $\Bar{\Delta}$ for realization $a$. Now, we perform balancing operation on this tree to produce a new realization $b$ with the same number of leaf nodes $\nint$. This can be done by removing the sibling pair at level $l_{max}^{a}$ and giving the leaf at $l_{min}^{a}$ a pair of children. Hence the tree:
\begin{enumerate}
    \item Loses one leaf at level $l_{min}^{a}$.
    \item Adds two leaves at level $l_{min}^{a}+1$.
    \item Loses two leaves at level $l_{max}^{a}$.
    \item Adds one leaf at level $l_{max}^{a}-1$.
\end{enumerate}
Therefore, 
\begin{equation}
    \Bar{\Delta}_{b} = \frac{1}{2}(1 + \frac{1}{\nint}\sum_{i=4}^{\nint}2^{l_i}+ \frac{1}{n}(2^{l_{max}^{a}-1}+2^{l_{min}^{a}+1}+2^{l_{min}^{a}+1})) .
    \label{eqn:unbalanced_no_2}
\end{equation}
Subtracting $\Bar{\Delta}_{b}$ from $\Bar{\Delta}_{a}$ and simplifying it further we get, 

\begin{equation}
    \Bar{\Delta}_{a} - \Bar{\Delta}_{b} = \frac{3}{2\nint}(2^{l_{max}^{a}-1}-2^{l_{min}^{a}}) \geq 0. 
    \label{eqn:difference_balance}
\end{equation}
With equality holding if $l_{min}^{a} = l_{max}^{a} - 1$, which is the case for a fully balanced \ac{FBT}. Hence, as long as the tree is not fully balanced, this balancing operation results in a lower $\Bar{\Delta}$. 
\end{proof}
%

We are interested in finding the least balanced realization that will provide the upper bound of $\Bar{\Delta}$ for a given number of users $\nint$. A fully unbalanced tree or a skewed tree, has two leaf nodes at $\nint-1$-th level and one leaf node in all the levels between $\nint-1$-th level and root node. Putting these values of $l_i$ in equation \eqref{eqn:aoi_for_network_scheudle} and using the expression for the sum a geometric series, we get the mean \ac{AoI} for a skewed tree,  
%



%
\begin{equation}
    \Bar{\Delta}_{skew} = \frac{1}{2}(1 + \frac{3\cdot2^{n}-1}{2\nint}).
    \label{eqn:aoi_for_network_scheudle_unbalanced}
\end{equation}

\subsection{Optimal selection of parameter J}
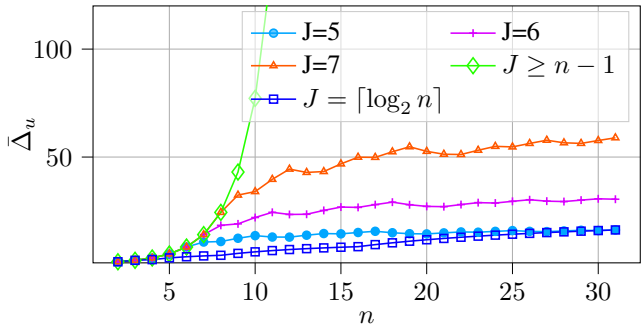
\begin{figure}[t]
    \centering
\begin{tikzpicture}
\begin{axis}[
width={\columnwidth}, height={5cm},
legend columns=2, 
legend cell align={left},
legend style={fill opacity=0.6, draw opacity=1, text opacity=1, draw=white!80!black},
tick align=outside,
tick pos=left,
x grid style={white!69.0196078431373!black},
xlabel={$\nint$},
xmajorgrids,
xmin=0.55, xmax=32.45,
xtick style={color=black},
y grid style={white!69.0196078431373!black},
ylabel={$\Bar{\Delta}_{u}$},
ylabel style={yshift=-0.8em},
ymajorgrids,
ymin=1, ymax=120,
ytick style={color=black}
]
\addplot [semithick, PlotColor1, mark=*, mark size=1.5pt]
table {%
2 1.5
3 2.16666666666667
4 3.25
5 5.1
6 8.33333333333333
7 10.6428571428571
8 10.875
9 12.3888888888889
10 13.6
11 12.9545454545455
12 12.9166666666667
13 13.8076923076923
14 14.5714285714286
15 14.4333333333333
16 15.0625
17 15.6176470588235
18 14.9444444444444
19 14.5
20 14.4
21 14.8809523809524
22 15.3181818181818
23 15.195652173913
24 15.5833333333333
25 15.94
26 15.5769230769231
27 15.462962962963
28 15.7857142857143
29 16.0862068965517
30 15.9666666666667
31 16.241935483871
};
\addlegendentry{J=5}
\addplot [semithick, PlotColor2, mark=+, mark size=1.5pt]
table {%
2 1.5
3 2.16666666666667
4 3.25
5 5.1
6 8.33333333333333
7 14.0714285714286
8 18.375
9 19.0555555555556
10 22
11 24.4090909090909
12 23.4166666666667
13 23.5
14 25.2857142857143
15 26.8333333333333
16 26.6875
17 27.9705882352941
18 29.1111111111111
19 27.9210526315789
20 27.15
21 27.0238095238095
22 28
23 28.8913043478261
24 28.7083333333333
25 29.5
26 30.2307692307692
27 29.5740740740741
28 29.3928571428571
29 30.051724137931
30 30.6666666666667
31 30.4677419354839
};
\addlegendentry{J=6}
\addplot [semithick, PlotColor3, mark=triangle, mark size=1.5pt]
table {%
2 1.5
3 2.16666666666667
4 3.25
5 5.1
6 8.33333333333333
7 14.0714285714286
8 24.375
9 32.3888888888889
10 34
11 39.6818181818182
12 44.4166666666667
13 42.8846153846154
14 43.2857142857143
15 46.8333333333333
16 49.9375
17 49.8529411764706
18 52.4444444444444
19 54.7631578947368
20 52.65
21 51.3095238095238
22 51.1818181818182
23 53.1521739130435
24 54.9583333333333
25 54.7
26 56.3076923076923
27 57.7962962962963
28 56.6071428571429
29 56.3275862068966
30 57.6666666666667
31 58.9193548387097
};
\addlegendentry{J=7}
\addplot [semithick, PlotColor4, mark=diamond, mark size=3pt]
table {%
2 1.5
3 2.16666666666667
4 3.25
5 5.1
6 8.33333333333333
7 14.0714285714286
8 24.375
9 43.0555555555556
10 77.2
11 140.045454545455
12 140.045454545455
13 140.038461538462
14 140.142857142857
15 140.83333333333
16 140.4375
17 140.02941176471
18 140.1111111111
19 140.0263157895
20 140.05
21 140.7380952381
22 140.090909091
23 140.02173913
24 140.458333333
25 140.42
26 140.07692308
27 140.68518519
28 140.89285714
29 140.0172414
30 140.0666667
31 140.016129
};
\addlegendentry{$J\geq\nint-1$}
\addplot [semithick, PlotColor5, mark=square, mark size=1.5pt]
table {%
2 1.5
3 2.16666666666667
4 2.5
5 3.3
6 3.83333333333333
7 4.21428571428571
8 4.5
9 5.38888888888889
10 6.1
11 6.68181818181818
12 7.16666666666667
13 7.57692307692308
14 7.92857142857143
15 8.23333333333333
16 8.5
17 9.44117647058824
18 10.2777777777778
19 11.0263157894737
20 11.7
21 12.3095238095238
22 12.8636363636364
23 13.3695652173913
24 13.8333333333333
25 14.26
26 14.6538461538462
27 15.0185185185185
28 15.3571428571429
29 15.6724137931034
30 15.9666666666667
31 16.241935483871
};
\addlegendentry{$J=\ceil{\log_{2}\nint}$}
\end{axis}

\end{tikzpicture}
    \caption{Selection of tree depth $J$: Forcing the maximum level to $J$ where $J<\nint-1$ stops the tree from settling into more unbalanced realizations. The lower the value of $J$, the better upper bound on AoI for a given $\nint$. However, we must be careful that $\nint$ never goes above $2^J$, as in that case there will not be enough available schedules for all the users and the PT will never settle.}
    \label{fig:pruning_tree}
\end{figure}

We can force the tree to settle in a way that more unbalanced realizations are possible to materialize. This can be done by selecting the parameter tree depth $J$ to be lower than $\nint-1$. As we reduce the value of $J$ further, we eliminate the possibility of the \ac{PT} settling into the more unbalanced realizations. Figure \ref{fig:pruning_tree} shows how the upper bound of the mean \ac{AoI} per user per slot can be decreased by decreasing $J$.   However, in order to have at least one schedule for each user, the tree depth must be greater than or equal to the height of a fully balanced realization , i.e., $J \geq \ceil{\log_{2}\nint}$. 


The selection of $J$ should take into account the maximum number of active users the system designer would like to provision for when the number of users is time-varying. In our system model presented in section \ref{sec:System_Model}, the number of active users will never exceed $M$. Therefore, we can safely select $J = \ceil{\log_{2}M}$. In general, the activation/deactivation model and its properties and total number of users in the network should be taken into consideration when selecting $J$.
%

%

\section{Evaluation}
\label{sec:Evaluation}

The box plots in Figure \ref{fig:resetlling_time} show the distribution of resettling time for a given $\nint= \{13,18,23,28\}$ under an activation or deactivation over 50 simulation runs. The upper and lower whisker of the boxplot encapsulate the entire range of the obtained data i.e the maximum and minimum resettling time. It is seen that the resettling time never exceeded 1100 slots and the maximum mean resettling time (center line of box plot) is 300 slots. Next, we simulate random geometrically distributed activations and deactivations for a system with $M=32$ users, $k=50,000$ and $n[0]=16$. The parameters for the algorithm shown in Table \ref{tab:params} were obtained using the gridsearch method. The channel utilization (fraction of successful slots) for 50 simulation runs of 50,000 slots each is shown in Figure \ref{fig:tpt_arrivals}. The maximum and minimum utilization over the 50 runs is shaded in the region around the mean. The number of users $\nint[t]$ is measured at the beginning of each batch of 100 slots. The seed of the random number generator which produced the activations and deactivations was kept the same to show the variation in the resettling time and to make meaningful comparisons between different schemes. The key demonstration from this figure is that the \ac{PT} is unsettled (observed by a drop in utilization) by a change in the number of active users. However, it then manages to \textit{always settle} and attain near full utilization after it is given enough time to resettle. Only the selected schedules of some users are disturbed during the resettling period which can be seen from the observation that the utilization in Figure \ref{fig:tpt_arrivals} never drops below 0.8. For comparison, \ac{SA} with the optimum access probability of $\frac{1}{\nint[t]}$ achieves utilization of only 0.4. \ac{SA} is a totally random medium access scheme and therefore suffers from collisions due to lack of coordination between users.

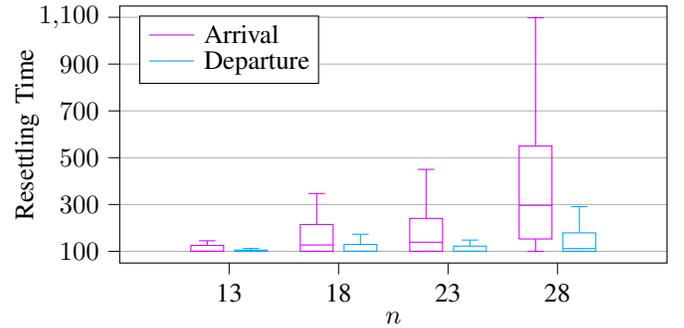
\begin{figure}[t]
    \centering
\begin{tikzpicture}

\definecolor{color1}{HTML}{00A6FF}
\definecolor{color0}{HTML}{D900FF}

\begin{axis}[
width={\columnwidth}, height={5cm},
legend cell align={left},
legend style={fill opacity=0.8, draw opacity=1, text opacity=1, at={(0.03,0.97)}, anchor=north west, draw=white!80!black},
tick align=outside,
tick pos=left,
x grid style={white!69.0196078431373!black},
xlabel={$\nint$},
xmin=-2, xmax=8,
xtick style={color=black},
xtick={0,2,4,6},
xticklabels={13,18,23,28},
y grid style={white!69.0196078431373!black},
ylabel={Resettling Time},
ymin=50.1, ymax=1147.9,
ytick style={color=black},
ytick={100,300,500,700,900,1100},
ymajorgrids,
]
\addplot [color0, forget plot]
table {%
-0.7 100
-0.1 100
-0.1 125
-0.7 125
-0.7 100
};
\addplot [color0, forget plot]
table {%
-0.4 100
-0.4 100
};
\addplot [color0, forget plot]
table {%
-0.4 125
-0.4 145
};
\addplot [color0, forget plot]
table {%
-0.55 100
-0.25 100
};
\addplot [color0, forget plot]
table {%
-0.55 145
-0.25 145
};
\addplot [color0, forget plot]
table {%
1.3 100
1.9 100
1.9 214.75
1.3 214.75
1.3 100
};
\addplot [color0, forget plot]
table {%
1.6 100
1.6 100
};
\addplot [color0, forget plot]
table {%
1.6 214.75
1.6 347
};
\addplot [color0, forget plot]
table {%
1.45 100
1.75 100
};
\addplot [color0, forget plot]
table {%
1.45 347
1.75 347
};
\addplot [color0, forget plot]
table {%
3.3 100
3.9 100
3.9 240.75
3.3 240.75
3.3 100
};
\addplot [color0, forget plot]
table {%
3.6 100
3.6 100
};
\addplot [color0, forget plot]
table {%
3.6 240.75
3.6 450
};
\addplot [color0, forget plot]
table {%
3.45 100
3.75 100
};
\addplot [color0, forget plot]
table {%
3.45 450
3.75 450
};
\addplot [color0, forget plot]
table {%
5.3 152.75
5.9 152.75
5.9 549.75
5.3 549.75
5.3 152.75
};
\addplot [color0, forget plot]
table {%
5.6 152.75
5.6 100
};
\addplot [color0, forget plot]
table {%
5.6 549.75
5.6 1098
};
\addplot [color0, forget plot]
table {%
5.45 100
5.75 100
};
\addplot [color0, forget plot]
table {%
5.45 1098
5.75 1098
};
\addplot [color1, forget plot]
table {%
0.1 100
0.7 100
0.7 105
0.1 105
0.1 100
};
\addplot [color1, forget plot]
table {%
0.4 100
0.4 100
};
\addplot [color1, forget plot]
table {%
0.4 105
0.4 112
};
\addplot [color1, forget plot]
table {%
0.25 100
0.55 100
};
\addplot [color1, forget plot]
table {%
0.25 112
0.55 112
};
\addplot [color1, forget plot]
table {%
2.1 100
2.7 100
2.7 129.5
2.1 129.5
2.1 100
};
\addplot [color1, forget plot]
table {%
2.4 100
2.4 100
};
\addplot [color1, forget plot]
table {%
2.4 129.5
2.4 173
};
\addplot [color1, forget plot]
table {%
2.25 100
2.55 100
};
\addplot [color1, forget plot]
table {%
2.25 173
2.55 173
};
\addplot [color1, forget plot]
table {%
4.1 100
4.7 100
4.7 122
4.1 122
4.1 100
};
\addplot [color1, forget plot]
table {%
4.4 100
4.4 100
};
\addplot [color1, forget plot]
table {%
4.4 122
4.4 148
};
\addplot [color1, forget plot]
table {%
4.25 100
4.55 100
};
\addplot [color1, forget plot]
table {%
4.25 148
4.55 148
};
\addplot [color1, forget plot]
table {%
6.1 100
6.7 100
6.7 178.75
6.1 178.75
6.1 100
};
\addplot [color1, forget plot]
table {%
6.4 100
6.4 100
};
\addplot [color1, forget plot]
table {%
6.4 178.75
6.4 291
};
\addplot [color1, forget plot]
table {%
6.25 100
6.55 100
};
\addplot [color1, forget plot]
table {%
6.25 291
6.55 291
};
\addplot [color0, forget plot]
table {%
-0.7 100
-0.1 100
};
\addplot [color0, forget plot]
table {%
1.3 127
1.9 127
};\label{p3}
\addplot [color0, forget plot]
table {%
3.3 139
3.9 139
};
\addplot [color0, forget plot]
table {%
5.3 296.5
5.9 296.5
};
\addplot [color1, forget plot]
table {%
0.1 100
0.7 100
};
\addplot [color1, forget plot]
table {%
2.1 100
2.7 100
};
\addplot [color1, forget plot]
table {%
4.1 100
4.7 100
};
\addplot [color1, forget plot]
table {%
6.1 112
6.7 112
};\label{p4}
\node [draw,fill=white!20] at (rel axis cs: 0.2,0.83) {\shortstack[l]{
\ref{p3} Arrival\\
\ref{p4} Departure}};
\end{axis}

\end{tikzpicture}
    \caption{Resettlig time: The resettling time after an arrival and departure on a settled tree with $J=5$. The mean, resettling time increases with the active number of users in the network $\nint$. The maximum resettling time never exceeds 1100 slots.}
    \label{fig:resetlling_time}
\end{figure}

Figure \ref{fig:final_result} shows the $\Bar{\Delta}$ for each batch for the same two cases. The shaded color marks the region between 10th and 90th percentile over 50 simulation runs. Here the $\Bar{\Delta}$ for \ac{SA} is greater than a factor of 4 as compared to mAQT. We also compare the analytical $\Bar{\Delta}$ of \ac{ADRA} \cite{ChenIoT}, where the users jointly optimise the channel access probability and an \ac{AoI} threshold based on $\nint[t]$. Here, the $\Bar{\Delta}$ for \ac{ADRA} is greater than a factor of 2 as compared to mAQT. Note that it was assumed that the users in \ac{SA} and \ac{ADRA} have a priori knowledge of $\nint[t]$ which was not the case for mAQT. This is an unrealistic advantage given to the users in \ac{SA} to demonstrate the power of our proposed method. Stationary random policies of \ac{SA} and \ac{ADRA} is far outperformed by mAQT due to the implicit coordination achieved between users over the feedback instead of relying on only on random chance.

\begin{table}[]
\centering
\begin{tabular}{|c|c|}
\hline
Algorithm & Weights      \\ \hline\hline
ALOHA-Q   & $2^{J}$     \\ \hline
ALOHA-QT      & $2^{J+1}-1$ \\ \hline
mAQT  & $2^{J+1}-1$ \\ \hline
\end{tabular}
\caption{The number of weights $|\mathcal{W}|$ required for the three expert-based RL algorithms. These weights are updated frequently and specify the amount of memory needed by each of the algorithms.}
\label{tab:memory}
\end{table}

In Figure \ref{fig:final_result2} we zoom in on the mAQT region and compare it with ALOHA-Q \cite{chu2015application}, which is another algorithm that employs \ac{RL} to achieve collision-free transmissions. In ALOHA-Q, each user selects a unique slot in a frame of fixed size $F$. A necessary condition for ALOHA-Q to settle is $\nint[t] \leq F$. We set the frame size to $2^J =F$, since that is the maximum number of users our setting for mAQT can accommodate. It is trivial to see that the upper bound for mAQT will always be smaller than that of ALOHA-Q since there will always be empty slots in ALOHA-Q unless $\nint[t] = 2^{J}$. From Tables \ref{tab:memory} and \ref{tab:complexity}, we see that the better performance of mAQT compared to ALOHA-Q comes at the cost of needing more memory and computation. The availability of different transmission rates in mAQT makes sure that no channel resources are wasted on idle slots, which is an important advantage over ALOHA-Q. The best case for our network setup would be for the users to transmit in a round-robin fashion \cite{faraziaoi}. This can be achieved via a centralized scheduling scheme or a partially centralized scheme with a complex feedback such as the DRR algorithm \cite{jiangdrr}. We show this RR plot as a baseline reference for the best case. To the best of out knowledge, no distributed \ac{MAC} algorithm gets closer the this baseline than mAQT. 

\begin{table}[]
\centering
\begin{tabular}{|c|c|c|c|}
\hline
       & ALOHA-Q & ALOHA-QT       & mAQT        \\ \hline\hline
Step 2 & $2^{J}$ & $2(2^{J+1}-1)$ & $2^{J+1}-1$ \\ \hline
Step 5 & $1$     & $J+1$          & $J+1$       \\ \hline
Step 6 & $1$     & $J+1$          & Skipped     \\ \hline
Step 7 & $2^{J}$ & $2^{J+1}-1$    & $2^{J+1}-1$ \\ \hline
Step 8 & $2^{J}$ & $2^{J+1}-1$    & $2^{J+1}-1$ \\ \hline
\end{tabular}
\caption{The worst case complexity ($\mathcal{O}$), of each step for the three expert-based RL algorithms. The proposed mAQT has half the run time in step 2 and skips step 6 entirely compared to ALOHA-QT. Its performance is the nonetheless better for our system model.}
\label{tab:complexity}
\end{table}

In Figure \ref{fig:final_result3}, we see that mAQT performs better than ALOHA-QT which justifies the modifications made to it, mentioned in section \ref{subsec:aloha_qt}. Table \ref{tab:complexity} shows the worst case run-time (complexity) of each step for ALOHA-QT and mAQT. The proposed scheme skips a threshold based search across all $2^{J+1}$ weights in step 2 and skips step 6. Additionally, the users in mAQT spend more than 50 percent of the 50,000 slots in a settled state. In mAQT, when the users are in a settled tree, they do run any of the complex steps from the Table \ref{tab:complexity} and only execute step 3. Hence, We obtain better performance at a lower cost by modifying ALOHA-QT for our system model.  

A final comparison of the $\Bar{\Delta}$ for all the 50,000 slots is shown in Table \ref{tab:final_result}. We see that mAQT shows a clear improvement over ALOHA-QT, even with the reduced complexity. Compared to ADRA, mAQT performs 50 percent better despite the fact that users in ADRA have additional knowledge of $\nint[t]$. 
\begin{figure}
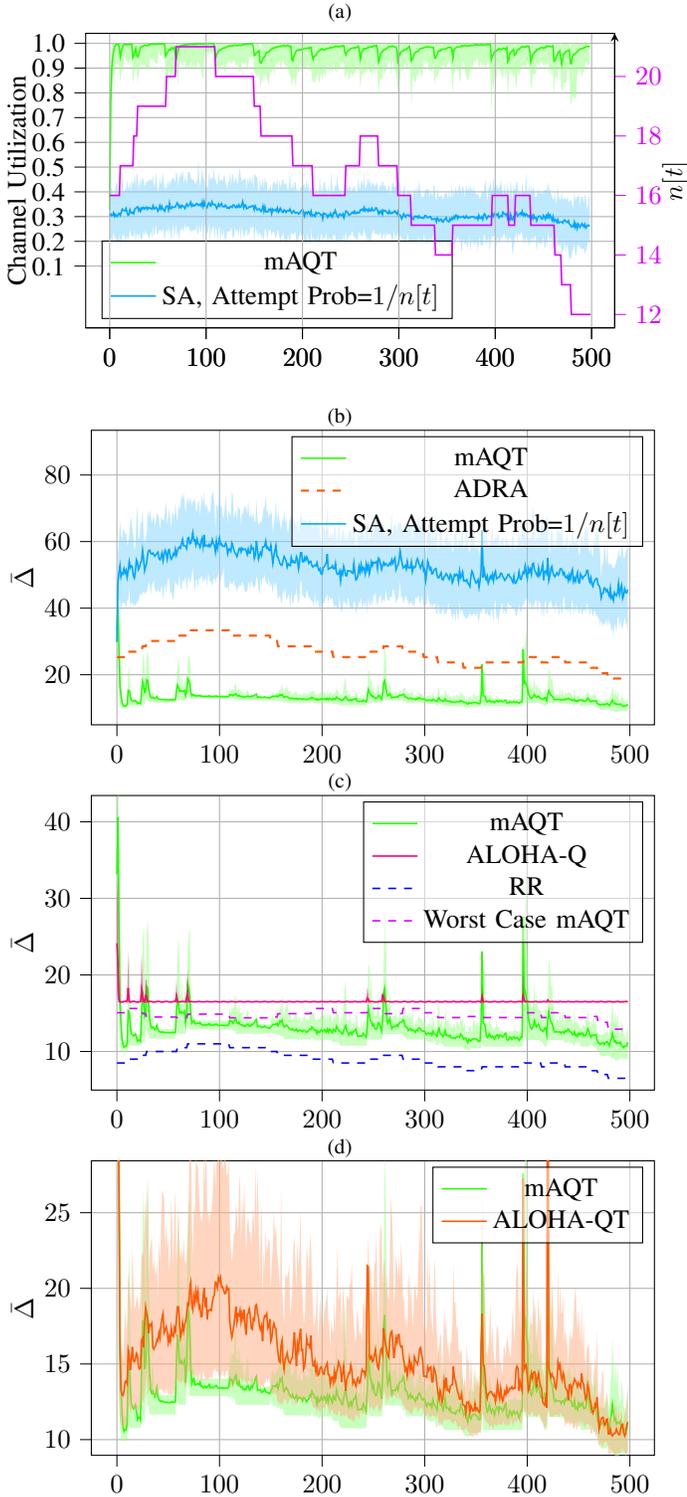

\begin{center}
\begin{subfigure}[b]{0.50\textwidth}
    \centering
    \caption{}
   \input{Images/Throughput_Arrivals}
   \label{fig:tpt_arrivals} 
\end{subfigure}
\begin{subfigure}[b]{0.50\textwidth}
    \centering
    \caption{}
   \input{Images/Final_Result}
   \label{fig:final_result}
\end{subfigure}
\begin{subfigure}[b]{0.50\textwidth}
    \centering
    \caption{}
    \input{Images/Final_result_1}
    \label{fig:final_result2}
\end{subfigure}
\begin{subfigure}[b]{0.50\textwidth}
    \centering
    \caption{}
    \input{Images/Final_result_2}
    \label{fig:final_result3}
    \end{subfigure}
\caption{Evaluation Results: The same arrival and departure pattern is applied for 30 simulation runs. (a) mAQT always settles and achieves two times higher utilization than SA. (b) The mean AoI comparison between mAQT, SA and ADRA. (c) When provisioning for the same number of users, mAQT performs better than ALOHA-Q. (d) The modifications suggested to ALOHA-QT improve the performance for the given system model.}
\end{center}
\end{figure}



\section{Conclusions}
\label{sec:conclusions}
\begin{table}[]
\centering
\setlength{\extrarowheight}{.5ex}
\begin{tabular}{|c|c|}
\hline
         & $\Bar{\Delta}$ \\ \hline\hline
RR       & $8.93$   \\ \hline
mAQT     & $13.07$  \\ \hline
ALOHA-QT & $15.32$  \\ \hline
ALOHA-Q  & $16.55$  \\ \hline
ADRA     & $26.71$  \\ \hline
SA       & $52.48$  \\ \hline
\end{tabular}
\caption{The mean network AoI per user per slot over all 50,000 slots. The proposed algorithm shows 15 percent reduction in $\Bar{\Delta}$ from ALOHA-QT and 50 percent reduction in the same from state of the art ADRA algorithm.}
\label{tab:final_result}
\end{table}
In this paper, we ponder the goal of minimizing the mean \ac{AoI} of a remote monitoring network with a time-varying number of users without a centralized scheduler. We make application-specific changes to the distributed \ac{RL} algorithm ALOHA-QT, which employs a policy tree (PT) to facilitate the coordination between users in a network so that they can select non-conflicting transmission slots. The users collectively obtain nearly full channel utilization when the \ac{PT} is settled. This settled \ac{PT} resembles a full binary tree, and the analysis of its properties shows that the balance of the settled tree affects the mean network \ac{AoI}. We also show how the selection of a design parameter in the algorithm, namely the tree depth $J$, can be used to improve the mean network \ac{AoI} by eliminating the possibility of the tree settling into more unbalanced realizations. Simulation results show that the suggested algorithm reduces mean network \ac{AoI} by 50 percent for state of the art age-dependent random access (ADRA) protocol, without the need for any interference cancellation or out-of-band communication. With this paper, we show that use of \ac{PT} to improve \ac{AoI} is promising in a decentralized \ac{MAC} setup. Some assumptions made in this work are not representative of real-life scenarios. For example, channel conditions other than interference, like noise, might cause a transmission or feedback to be lost. Hence, our future work will include implementing mAQT on a hardware testbed such as \cite{ayan2021} to investigate its performance outside of simulations.

\section*{Acknowledgment}
This work has been carried out with the support of DFG priority programme Cyber-Physical Networking (CPN) with the grant number KE 1863/5-2 and the Federal Ministry of Education and Research of Germany (BMBF) programme of "Souverän. Digital. Vernetzt." joint project 6G-life with project identification number 16KISK002. The authors would also like to thank  Dr. Ph.D. Fidan Mehmeti for his valuable inputs.

\bibliographystyle{IEEEtran}

\bibliography{bibliography}

\end{document}